\documentclass[article]{IEEEtran}
\pdfoutput=1
\usepackage[utf8]{inputenc}
\usepackage{dsfont}

\usepackage[compatible]{nomencl}

\makenomenclature

\makeindex

\usepackage{richard_config}

\ifCLASSOPTIONcompsoc
	\usepackage[noadjust,nocompress]{cite}
\else
	\usepackage[noadjust]{cite}
\fi

\ifCLASSINFOpdf
  \usepackage[pdftex]{graphicx}
  \graphicspath{{./figures/}}
\else
\fi

\author{\IEEEauthorblockN{Abdoulaye Tall\IEEEauthorrefmark{1}, Richard Combes\IEEEauthorrefmark{2}, Zwi Altman \IEEEauthorrefmark{1} and Eitan Altman\IEEEauthorrefmark{3}} \\ \IEEEauthorblockA{\IEEEauthorrefmark{1}Orange Labs
38/40 rue du General Leclerc,92794 Issy-les-Moulineaux \\Email: \{abdoulaye.tall,zwi.altman\}@orange.com}\\ \IEEEauthorblockA{\IEEEauthorrefmark{2}KTH, Royal Institute of Technology, Stockholm, Sweden\\ Email: rcombes@kth.se}\\ \IEEEauthorblockA{\IEEEauthorrefmark{3}INRIA Sophia Antipolis, 06902 Sophia Antipolis, France\\Email:eitan.altman@sophia.inria.fr}
\thanks{This work has been partially carried out in the framework of the FP7 UniverSelf project under EC Grant agreement 257513}}

\title{Distributed coordination of self-organizing mechanisms in communication networks}
\begin{document}
\maketitle
\begin{abstract}
The fast development of the \ac{SON} technology in mobile networks renders the problem of coordinating \ac{SON} functionalities operating simultaneously critical. \ac{SON} functionalities can be viewed as control loops that may need to be coordinated to guarantee conflict free operation, to enforce stability of the network and to achieve performance gain.
This paper proposes a distributed solution for coordinating \ac{SON} functionalities. It uses Rosen's concave games framework in conjunction with convex optimization. The \ac{SON} functionalities are modeled as linear \ac{ODE}s. The stability of the system is first evaluated using a basic control theory approach. The coordination solution consists in finding a linear map (called coordination matrix) that stabilizes the system of \ac{SON} functionalities. It is proven that the solution remains valid in a noisy environment using Stochastic Approximation. A practical example involving three different \ac{SON} functionalities deployed in \acp{BS} of a \ac{LTE} network demonstrates the usefulness of the proposed method.


\end{abstract}

\nomenclature{$det(\bullet)$}{Determinant of a matrix}
\nomenclature{$\textbf{tr}(\bullet)$}{Trace of a matrix}
\nomenclature{$A \prec 0$}{$A$ is negative definite}
\nomenclature{$A^T$}{Transpose of matrix $A$}
\nomenclature{$eig(A)$}{Eigenvalues of $A$}
\nomenclature{$\cal C$}{Set of matrices having the same particular form (e.g., with zero elements at specific positions)}
\nomenclature{$[A]_{k,k}$}{$k^{th}$ order leading principal submatrix of $A$}
\nomenclature{$\nabla_x f$}{Gradient of $f$ with regard to $x$}
\nomenclature{$JF(\bullet)$}{Jacobian of $F()$}
\nomenclature{$\mathds{1}_{\{\textit{cond}\}}(x)$}{Indicator function on the set of $x$ values satisfying the condition $\textit{cond}$}
\nomenclature{$\mathbb{E}(\bullet)$}{Expectation of a random variable}
\nomenclature{$\dot{\alpha}$}{Derivative of $\alpha$ over time}
\nomenclature{$[\bullet]_S^+$}{Projection on the set $S$}
\nomenclature{$\bar{x}$}{Notation for constants often representing target values}
\nomenclature{$\theta$}{Vector of parameters}
\nomenclature{$\theta^*$}{Equilibrium points of system comprizing control loops}
\nomenclature{$A$}{Real matrix representing linearized system of \ac{SON} functions}

\section{Introduction}

The \ac{RAN} landscape is becoming increasingly complex and heterogeneous with co-existing and co-operating technologies. \ac{SON} mechanisms have been introduced as a means to manage complexity, to reduce cost of operation, and to enhance performance and profitability of the network. Self-organizing networks enable automation of the configuration of newly
deployed network nodes (self-configuration), automation of parameter tuning for \acp{KPI} improvement (self-optimization) and also automation of diagnostic and reparation of faulty network nodes (self-healing). In \ac{LTE} networks, large scale deployment of the \ac{SON} functionalities has started with self-configuration \ac{SON}s to simplify the network deployment, and that of self-optimization functionalities is expected to follow.

Self-optimization mechanisms can be viewed as control loops, that can be deployed in the management or the control plane. The former is often denoted as  centralized \ac{SON} and the latter as distributed \ac{SON}. In the centralized case, the \ac{SON} algorithms are deployed in the \ac{OMC} or in the \ac{NMS} which are part of the \ac{OSS}. Centralized \ac{SON} benefits from abundant data (metrics and \acp{KPI}) and computational means necessary for processing and running powerful optimization methods \cite{LuocentSON,SugacentSON}. The main drawback of the centralized approach is related to the long time scale that is typically used, in the order of an hour and more. Hence the \ac{SON} algorithms cannot adapt the network to traffic variations that occur at short time scales.

The second approach, namely the distributed \ac{SON}, is more scalable since the optimization is performed locally involving one or several \acp{BS}. The main advantage of the distributed \ac{SON} is its higher reactivity, namely its ability to track quick variations in the propagation conditions or in the traffic \cite{CombesInfocom2012,rze2013icic} and to adapt system parameters in the time scale of seconds (i.e. flow level duration). The higher reactivity sometimes impacts the type of solution sought, namely a solution which targets local minima instead of global minima. However, distributed optimization can also reach global minima \cite{CombesInfocom2012}.

\ac{SON}s are often designed as stand alone functionalities, and when triggered simultaneously, their interactions are not always predictable. The deployment of multiple control loops raises the questions of conflicts, stability and performance. The topics of conflict resolution, coordination, and a framework for managing multiple \ac{SON} functionalities are receiving a growing interest (see for example \cite{SEMAFOUR2009,SEMAROUR2011, UNIVERSELF2013}). Most contributions that have addressed the coordination problem  between specific \ac{SON} functionalities, provide a solution implemented in a centralized \cite{LiuMROMLB,klessig2012improving} or distributed \cite{RICHARD_RelaysTNSM} fashion. In a centralized solution, the  \ac{SON} coordination can be treated as a multi-objective optimization \cite{VlacheasSONcoord}. From the standardization point of view, the coordination problem has been addressed as a centralized, management-plane problem \cite{3gpp.28.628}.

Little material has been reported on distributed, control plane solutions for the coordination problem in spite of its higher reactivity, attractiveness from an architecture point of view and potential performance gain. The aim of this paper is to provide a generic coordination mechanism which is practically implementable. The contributions of the paper are the following:
\begin{itemize}
\item The problem of SON coordination is analyzed using a control theory/stochastic approximation-based framework.
\item The case of fully distributed coordination is addressed.
\item It is shown that coordination can be formulated as a convex optimization problem.
\item The coordination solution is applied to a use case involving 3 \ac{SON} functions deployed in several \acp{BS} of a wireless network.
\end{itemize}

	A first version of this paper has already been published in \cite{rze2012coordWiOpt}. New results presented here include the formulation of the coordination problem as a convex optimization problem with \ac{LMI} constraints: stability constraint and connectivity constraints related to the capability of the self-organizing nodes to exchange information via the transport network. The merit of the proposed solution is its capability to handle a large number of control loops and enforce their stability, as illustrated in the use case of \ac{SON} deployment in a \ac{LTE} network. To our knowledge, this is the first generic control plane solution to the problem of SON coordination in a mobile network.

	The paper is organized as follows: in Section~\ref{sec:parallel_control_loops} we state the proposed model for interaction of \ac{SON} mechanisms running in parallel and the coordination problem to be solved. In Section~\ref{sec:coordinationgen} we focus on the case where performance indicators are affine functions of the parameters, and propose a practically implementable coordination mechanism. In Section~\ref{sec:full distributed} we study fully distributed coordination when no exchange of information between \acp{SON} is needed, and we show that it is not always possible. In Section~\ref{sec:coordmeth} the coordination problem is formulated as a convex optimization problem which can be fast solved with modern computers. In Section~\ref{sec:application} we illustrate the application of our model to traffic management in wireless networks with two examples including a use case of the coordination method in a \ac{LTE} network with three different \ac{SON} functionalities deployed in each \ac{BS}. Section~\ref{sec:conclusion} concludes the paper. In appendices~\ref{app:odes} and~\ref{app:linear_ode} we recall the basic notions of stability for \acp{ODE} and linear \acp{ODE} respectively.
	
\printnomenclature

\section{Problem Description}\label{sec:parallel_control_loops}
\subsection{\ac{SON} model}
	A \ac{SON} mechanism is an entity which monitors a given performance indicator and controls a scalar parameter. The current value of the performance indicator is observed, and the parameter is modified accordingly to attain some objective. We consider $I > 1$ \ac{SON} mechanisms operating simultaneously. We define $\theta_i$ the parameter controlled by the $i$-th \ac{SON} mechanism and $\theta = (\theta_1,\dotsc,\theta_I)$ the vector of parameters. The $i$-th \ac{SON} mechanism monitors a performance indicator $F_i(\theta)$ and updates its parameter $\theta_i$ proportionally to it. $F(\theta) = (F_1(\theta),\dotsc,F_I(\theta))$ is the vector of update of $\theta$.

	We say that the $i$-th \ac{SON} mechanism operates in \emph{stand-alone} mode if all parameters but $\theta_i$ are kept constant. The $i$-th \ac{SON} mechanism operating in stand-alone is described by the \ac{ODE}:
\begin{align}\label{eq:son_standalone}
	\dot{ \theta_i } = F_i(\theta) \;,\; \dot{ \theta_j } = 0 , j \neq i.
\end{align}
	We say that the \ac{SON} mechanisms operate in \emph{parallel} mode if all parameters are modified simultaneously, which is described by the \ac{ODE}:
\begin{align}\label{eq:son_simultaneous}
	\dot{ \theta } &= F(\theta).
\end{align}
 	We say that the $i$-th \ac{SON} mechanism is stable in stand-alone mode if there exists $\theta^{*,i}_i$ for fixed $\theta_j$, $j \neq i$ which is \emph{asymptotically stable} for~\eqref{eq:son_standalone}. The definition of asymptotic stability is recalled in appendix~\ref{app:odes}. It is noted that $\theta^{*,i}_i$ depends on $\theta_j$, $j \neq i$. We say that the \ac{SON} mechanisms are stable in parallel mode if there exists $\theta^*$ which is asymptotically stable for~\eqref{eq:son_simultaneous}. Typically, the \ac{SON} mechanisms are designed and studied in a stand-alone manner, so that stand-alone stability is verified. 
	
	However, stand-alone stability does not imply parallel stability. First consider a case where $F_i(\theta)$ does not depend on $\theta_j$, for all $j \neq i$. Then~\eqref{eq:son_simultaneous} is a set of $I$ parallel independent \acp{ODE}, so that stand-alone stability implies parallel stability. On the other hand, if there exists $i \neq j$ such that $F_i(\theta)$ depends on $\theta_j$, then the situation is not so clear-cut. We say that \ac{SON} $i$ and $j$ \emph{interact}. Namely, interaction potentially introduces \emph{instablity}.
	
	In the remainder of this article we will be concerned with conditions for parallel stability, and designing coordination mechanisms to force stability whenever possible.

\subsection{Stability characterization}\label{subsec:stabchar}
	Two particular cases of parallel mechanisms will be of interest. The first case is what we will call \emph{zero-finding} algorithms. Each \ac{SON} mechanism monitors the value of a performance indicator and tries to achieve a fixed target value for this performance indicator. Namely:
\begin{equation}\label{eq:son_target}
	F_i(\theta) =  f_i(\theta) - \overline{f}_i,
\end{equation}
	where $f_i$ is the performance indicator monitored by \ac{SON} $i$ and $\overline{f}_i$ - a target level for this performance indicator. The goal of the $i$-th \ac{SON} mechanism is to find $\theta_i^*$ for fixed $\theta_j$, $j \neq i$ such that $f_i(\theta_1,\dotsc,\theta_i^*,\dotsc,\theta_I) = \overline{f}_i$. If $\theta_i \mapsto f_i(\theta_1,\dotsc,\theta_i,\dotsc,\theta_I)$ is strictly decreasing $1 \leq i \leq I$ then stand-alone stability is assured. Indeed, $V_i(\theta) = (f_i(\theta) - \overline{f}_i)^2$ would be a Lyapunov function for \eqref{eq:son_target}.

	Another case of interest is maximization algorithms. Each \ac{SON} mechanism tries to maximize a given performance indicator. There exists a continuously differentiable function $g_i$ such that:
\begin{equation}\label{eq:son_maximize}
	F_i(\theta) =  \nabla_{\theta_i} g_i(\theta).
\end{equation}
	In stand-alone mode, \ac{SON} $i$ indeed converges to a local maximum of $\theta_i \to g_i(\theta)$. If we restrict $\theta$ to a closed, convex and bounded set and if $\theta_i \mapsto g_i(\theta_1,\dotsc,\theta_i,\dotsc,\theta_I)$ is concave $1 \leq i \leq I$ , we fall within the framework of \emph{concave games} considered in \cite{RosenConvexGame}.
	Note that zero-finding algorithms can be rewritten as maximization algorithms by choosing $g_i(\theta) = -(f_i(\theta) - \overline{f}_i)^2$.
	
	An important result of \cite{RosenConvexGame} for parallel stability is given in the following theorem. Denote by $w \in \mathbb{R}_+^I$ a vector of \emph{real positive weights}.
	\begin{theorem} \label{th:diagstrictconc}
		 If $\sum_{i=1}^I w_i g_i(\theta)$ is diagonally strictly concave for $\theta$ in a convex bounded set $S \subset \mathbb{R}^I$, then the system of \acp{ODE} $\dot{\theta} = w^T \cdot F(\theta)$ admits a unique equilibrium point that is asymptotically stable in $S$.
	\end{theorem}
	If we denote by $J_{F,w}$ the Jacobian of $w \cdot F(\theta) = [w_1 F_1(\theta),...,w_I F_I(\theta)]$, a sufficient condition for diagonal strict concavity is that $J_{F,w} + J_{F,w}^T$ is negative definite.

	Note that \eqref{eq:son_simultaneous} is a special case of the \acp{ODE} considered in Theorem \ref{th:diagstrictconc} with $w_i = 1, i=1, \dots, I$. Without diagonal strict concavity there is no guarantee that parallel stability occurs, and coordination is needed.
	
\subsection{Linear case}
	In the remainder of this paper, we study the case where $F$ is affine:
\begin{equation}\label{eq:son_linear}
	F(\theta) = A \theta + b,
\end{equation}
with $b$ a vector of size $I$ and $A$ a matrix of size $I \times I$. We assume that $A$ is invertible and we define $\theta^* = -A^{-1} b$. 
	The \ac{SON} mechanisms running in parallel are described by the linear \ac{ODE}:
\begin{equation}\label{eq:son_linear_ode}
	\dot{\theta} = A \theta + b = A (\theta - \theta^*).
\end{equation}
	It is noted that in the linear case, we always fall within the scope of zero-finding algorithms described previously, by defining:
\begin{align}\label{eq:son_linear_ode_zero}
	f_i(\theta) =  \sum_{1 \leq j \leq I} A_{i,j} \theta_j \;,\; \overline{f}_i &= -b_i, \\
	\dot{\theta}_i = f_i(\theta) - \overline{f}_i.
\end{align}
	In particular, stand-alone stability occurs \ac{iff} $A_{i,i} < 0$ , $1 \leq i \leq I$, i.e all the diagonal terms of $A$ are strictly negative. Basic results on linear \acp{ODE} are recalled in appendix~\ref{app:linear_ode}. Namely, parallel stability holds \ac{iff} all the eigenvalues of $A$ have a strictly negative real part. We then say that $A$ is a \emph{Hurwitz matrix}.
	
	For practical systems, performance indicators $F(\theta)$ need not be linear functions of $\theta$. However, as long as they are smooth, they can be approximated by linear functions using a Taylor expansion in the neighborhood of a stationary point assuming a Taylor expansion exists. Consider a point $\theta^*$ with $F(\theta^*) = 0$. If the values of $\theta$ are restricted to a small neighborhood of $\theta^*$:
\begin{equation} 
	F(\theta) \approx J F(\theta^*) (\theta - \theta^*),
\end{equation}
with $J F(\theta^*)$ the Jacobian of $F$ evaluated at $\theta^*$. The Hartman-Grossman theorem (\cite{HartmanGrossman}) states that on a neighborhood of $\theta^*$, stability of the system with linear approximation implies stability of the original, non-linear system. Hence implementing the proposed coordination mechanism where $A$ is replaced by $J F$ ensures stability if we constrain $\theta$ to a small neighborhood of $\theta^*$.
	
	The parameters $A$ and $b$ might be unknown, and we can only observe noisy values of $F(\theta)$ for different values of $\theta$. The crudest approach is to estimate $A$ and $b$ through finite differences:
	\begin{align}
	a_{i,j} &\approx \frac{ f_j(\theta +  e_i \delta \theta_i) - f_j(\theta -  e_i \delta \theta_i)  }{ 2 \delta \theta_i }, \\
	b_{i} &\approx f_i(0).
	\end{align}
	with $e_i$ the $i$-th unit vector. The results are averaged over several successive measurements and additive noise is omitted for notation clarity. In general, the measurements of $F$ are obtained by calculating the time average of some output of the network during a relatively long time, so that a form of the central limit theorem applies and the additive noise is Gaussian. In this case, a better method is to employ \emph{least-squares regression}. Least-squares regression is a well studied topic  with very efficient numerical methods (\cite{LeastSquaresBjorck}) even for large data sets so that estimation of $A$ and $b$ is not computationally difficult.
	
	Finally, since practical systems do not remain stationary for an infinite amount of time, a database with values of $A$ and $b$ for each set of operating conditions must be maintained. In the context of wireless networks, the relationship between parameters and performance indicators changes when the traffic intensity changes because of daily traffic patterns. For instance, during the night traffic is very low, and traffic peaks are observed at the end of the day. A database with estimated values of $A$ and $b$ at (for instance) each hour of the day should be constructed.
	
\section{Coordination} \label{sec:coordinationgen}
\subsection{Coordination mechanism}
	If $A$ has at least one eigenvalue with positive or null real part, convergence to $\theta^*$ does not occur, and a coordination mechanism is needed. We consider a \emph{linear} coordination mechanism, where $A$ is replaced by $C A$ with $C$ a $I \times I$ real matrix. The \ac{ODE} for the coordinated system is:
\begin{equation}\label{eq:son_linear_ode_coord}
	\dot{\theta} = C A (\theta - \theta^*).
\end{equation}
	The coordination mechanism can be interpreted as transforming the performance indicator monitored by \ac{SON} $i$ from $f_i$ to a linear combination of the performance indicators monitored by all the \ac{SON} mechanisms. 
	 As explained in appendix~\ref{app:linear_ode}, stability is achieved if there exists a symmetric matrix $X$ such that: \begin{align}
	(C  A)^T X + X C A \prec 0 \;\;,\;\;	X \succ 0,
\end{align}
where $X \succ 0$ denotes that $X$ is positive definite. In particular, \begin{equation} 
V(\theta) = (\theta - \theta^*)^T X (\theta - \theta^*),
\end{equation} 
acts as a Lyapunov function.
\subsection{Distributed implementation}	
	The choice for the coordination matrix $C$ is not unique. For instance $C = -A^{-1}$ ensures stability. For the coordination mechanism to be scalable with respect to the number of \acp{SON}, $C$ should be chosen to allow \emph{distributed} implementation. We say that \ac{SON} $j$ is a neighbor of \ac{SON} $i$ if $\frac{\partial f_j}{\partial \theta_i} \neq 0$. We define ${\cal I}_i$ the set of neighbors of $i$. The coordination mechanism is distributed if each \ac{SON} needs only to exchange information with its neighbors.
	
	We give an example of a coordination mechanism which can always be distributed. The mechanism is based on a \emph{separable} Lyapunov function. Define the weighted square error:
\begin{equation}
	V(\theta) = \sum_{i=1}^I w_i (f_i(\theta) - \overline{f}_i)^2 = (\theta - \theta^*)^T A^T W A (\theta - \theta^*),
\end{equation}
with $W = \diag(w)$ i.e the diagonal matrix with diagonal elements $\Set{w_i}_{1 \leq i \leq I}$.
	Coordination is achieved by following the gradient of $-V$ so that $V$ is a Lyapunov function:
\begin{equation}
	\dot{\theta} = - \nabla_{\theta} V(\theta) =  - A^T W A (\theta - \theta^*).
\end{equation}
	Namely, we choose $C = - A^T W$. We can verify that the mechanism is distributed:
\begin{equation}
	\dot{\theta}_i =  \sum_{j=1}^I  2 w_i  \frac{\partial f_j}{\partial \theta_i}  (  f_j(\theta) - \overline{f}_j ) =  \sum_{j \in {\cal I}_i} 2 w_j  \frac{\partial  f_j}{\partial  \theta_i}  (  f_j(\theta) - \overline{f}_j ).
\end{equation}
	Indeed, the update of $\theta_i$ only requires knowledge of $\frac{\partial f_j}{\partial \theta_i}$ and $f_j(\theta) - \overline{f}_j$, for $j \in {\cal I}_i$, and this information is available from the \emph{neighbors} of $i$.
	
	It is also noted that such a mechanism can be implemented in an \emph{asynchronous manner}, i.e the components of $\theta$ are updated in a round-robin fashion, or at random instants, and the average frequency of update is the same for all components. The reader can refer to \cite{Bertsekas}[chapters 6-8] for the round-robin updates  and \cite{Kushner}[chapter 12] for the random updates. Asynchronous implementation is important in practice because if the \acp{SON} are not co-located, maintaining clock synchronization among the \acp{SON} would generate a considerable amount of overhead.
\subsection{Stochastic Control Stabilization}
	In practical systems, \acp{ODE} are replaced by stochastic approximation algorithms. Indeed, the variables are updated at discrete times proportionally to functions values which are noisy. 
	
	The noise in the function values is due to the fact that time is slotted and functions are estimated by averaging certain counters during a time slot. For example, the load of a \ac{BS} in a mobile network is often estimated by evaluating the proportion of time during which the scheduler is busy, and the file transfer time is estimated by averaging the file transfer times of all flows occurring in a certain period of time. The noise is also due to intrinsic stochastic nature of real systems, for example in wireless networks the propagation environment is inherently non-deterministic (because of fading, mobility, etc.) so the \ac{SINR} will be noisy. 
	
	When the noise in the measurements of the function values is of Martingale difference type, the mean behavior of those \ac{SA} algorithms matches with the system of \acp{ODE}. Note that we consider Martingale difference type of noise but the SA results hold for much more general noise processes (stationary, ergodic). In \cite{RICHARD_RelaysTNSM} for example, \ac{SA} results are used without the Martingale difference property.
	
	The initial system of control loops is in reality a system of \ac{SA} algorithms, with one of them written as
\begin{equation}\label{s.a.simple}
\theta_i[k+1] = \left[\theta_i[k] + \epsilon_k(f_i(\theta[k]) + N_k^i) \right]_{S_i}^+
\end{equation}
where $[.]_{S_i}^+$ is the projection on the interval $S_i = [a_i,b_i]; a_i<b_i \in \mathbb{R}$, $\theta[k] = (\theta_1[k],...,\theta_I[k])$ is the vector of parameters after the ($k-1$)th update, $\epsilon_k$ the step of the $k$th update and $N_k^i$ represents the noise in the measurement.

	The projection in \eqref{s.a.simple} aims at ensuring that the iterates are bounded. This is also a condition for convergence of the \ac{SA} algorithm towards the invariant sets of the equivalent \ac{ODE}.
	
	Most \ac{SON} algorithms are or can be reduced to the form of \eqref{s.a.simple}. For example in \cite{CombesInfocom2012}, a load balancing \ac{SON} is presented in this very same form. In \cite{RICHARD_RelaysTNSM} relays are self-optimized using also a \ac{SA} algorithm. In \cite{yun2010ctrl}, \ac{SA} algorithms are used for self-optimizing interference management for femtocells. A handover optimization \ac{SON} which can be rewritten as an \ac{SA} algorithm is also presented in \cite{viering2009mathematical}.

	We suppose that $N_k$ is a martingale difference sequence to meet the conditions for stand alone convergence (see \cite{Borkar,Kushner}). Namely the \ac{SA} algorithms have the same behavior as their equivalent \ac{ODE}. Now we want to check if the conditions for the \ac{SA} equivalence with the limiting \ac{ODE} are still verified after the coordination mechanism is applied. The coordinated \ac{SA} for the $i$-th mechanism is
\begin{equation}\label{s.a.coord}
\theta_i[k+1] = \left[\theta_i[k] + \epsilon_k(\sum_{j=1}^I C_{i,j} (f_j(\theta[k]) + N_k^j)) \right]_{S_i}^+
\end{equation}

	The projection ensures that the iterates are bounded. The question now is to show that $\sum_{j=1}^I C_{i,j} N_k^j$ is a Martingale difference sequence in order to meet the convergence conditions. Denoting ${\cal F}_k = \left\lbrace \sum_{j=1}^I C_{i,j} N_l^j, l<k \right\rbrace$, we have
\begin{align*}
E \left[ \sum_{j=1}^I C_{i,j} N_k^j | {\cal F}_k \right] = 
 \sum_{j=1}^I C_{i,j} E \left[ N_k^j | {\cal F}_k \right] = 0
\end{align*}
since $E[ N_k^j | N_l^j, l<k] = 0, j=1...I$. So this condition is satisfied ensuring the validity of the coordination method in a stochastic environment.

\section{Fully distributed coordination}\label{sec:full distributed}
	In this section we study fully distributed coordination, where the coordination matrix $C$ is \emph{diagonal}. As said previously, if $C_{i,j} \neq 0$, $i \neq j$ then \ac{SON} $i$ and $j$ need to exchange information. In fully distributed coordination, no information is exchanged. We prove two results. For $I = 2$ fully distributed coordination can always be achieved. For $I=3$ it is also possible if $A^{-1}$ has at least one non-zero diagonal element and impossible otherwise. 	These results are attractive from a practical point of view: if there are $3$ or less \acp{SON} to coordinate, it suffices to modify their feedback coefficient, without any exchange of information or interface between them. We say that the system can be coordinated in a \emph{fully distributed} way \ac{iff} there exists $c \in \RR^I$ such that $\diag(c) A$ is a Hurwitz matrix.
	
	The following lemma will be useful. It is a consequence of the Routh-Hurwitz theorem (\cite{gantmacher2005applications}).
	\begin{lemma}\label{lem:routh}
	Let $M$ a $I \times I$ invertible real matrix. For $I=2$ , $M$ is a Hurwitz matrix \ac{iff}
	\begin{equation} \det(M) > 0 \;,\; \tr(M) < 0\end{equation}
where $\tr$ denotes the trace of a matrix.

	For $I=3$ , $M$ is a Hurwitz matrix \ac{iff} \begin{equation} \det(M) < 0 \;,\; \tr(M) < 0 \;,\;  \tr(M) \tr(M^{-1})> 1.\end{equation}
	\end{lemma}
	For $I = 2$ mechanisms, the system can always be coordinated in a fully distributed way as shown by Theorem~\ref{th:coord2}.
	\begin{theorem}\label{th:coord2}
		For $I=2$, the system can always be coordinated in a fully distributed way. $\diag(c) A$ is a Hurwitz matrix \ac{iff}:
		\begin{equation} c \in  {\cal C} = \Set{ c  :   c_1 A_{1,1} + c_2 A_{2,2} < 0  ,  c_1 c_2 \det(A) > 0 },\end{equation} and ${\cal C}$ is not empty since:
	\begin{equation} \left(1, \sign \det(A) \frac{\abs{A_{1,1}}}{2 \abs{A_{2,2}}}\right) \in {\cal C}.\end{equation}
	\end{theorem}
	\begin{proof}
		 $\tr( \diag(c) A) = c_1 A_{1,1} + c_2 A_{2,2}$ and $\det( \diag(c) A) =  c_1 c_2 \det(A)$.	Using Lemma~\ref{lem:routh} proves the first part of the result. ${\cal C}$ is not empty, since one of its elements is given by inspection of the proposed value.
	\end{proof}	
	For $I = 3$ mechanisms, the system can also be coordinated in a fully distributed way providing that the inverse of $A$ has one non-zero diagonal element as shown by Theorem~\ref{th:coord3}.
\begin{theorem}\label{th:coord3}
		For $I=3$, the system can be coordinated in a fully distributed way if $B = A^{-1}$ has at least one non-zero diagonal element. Assume that $B_{2,2} \neq 0$ without loss of generality. Consider $\epsilon > 0$, and define $C(\epsilon) = \diag(1 , \epsilon c_2 , \epsilon c_3 )$ with:
		\begin{equation} \frac{B_{2,2}}{c_2} + \frac{B_{3,3}}{c_3} < 0 \;,\; c_2 c_3 \det(A) < 0.\end{equation}
		A possible choice for $(c_2,c_3)$ is $c_2 =  - B_{2,2}$ and $c_3 = -2 \sign( \det(A) c_2) \abs{B_{3,3}}$ if $B_{3,3} \neq 0$ and $c_3 = - \sign( \det(A) c_2)$ otherwise.

		Then there exists $\epsilon_0$ such that $ C(\epsilon) A$ is a Hurwitz matrix for $0 < \epsilon < \epsilon_0$.
		
		If $B$ has a null diagonal, the system cannot be coordinated.
\end{theorem}
\begin{proof} We have that
	\begin{align*}\tr( C(\epsilon) A ) & = A_{1,1} + O(\epsilon), \\
	 \tr( C(\epsilon) A )  \tr( ( C(\epsilon) A)^{-1} ) &= \frac{A_{1,1}}{\epsilon} \left( \frac{B_{2,2}}{c_2} + \frac{B_{3,3}}{c_3} \right)    + O(1).\end{align*}
	 Recall that $A_{1,1} < 0$. So there exists $\epsilon_0$ such that $\tr( C(\epsilon) A ) < 0$ and $\tr( C(\epsilon) A )  \tr( ( C(\epsilon) A)^{-1} ) > 1$, if $\epsilon > \epsilon_0$. Using Lemma~\ref{lem:routh},  $C(\epsilon) A$ is a Hurwitz matrix for $0 < \epsilon < \epsilon_0$. The existence of a couple $(c_2,c_3)$ is given by inspection of the proposed value. If $B$ has a null diagonal, then $\tr( ( \diag(c) A)^{-1} ) = 0$ for all $c$, so that the conditions of Lemma~\ref{lem:routh} can never be met.
	\end{proof}
	
	For $I > 3$, the problem becomes more involved. Sufficient conditions for the existence of a diagonal matrix can be found in the literature. In particular Fisher and Fuller (1958) \cite{fisher1958stabilization} have proven that if there exists a permutation matrix P such that all leading principal sub-matrices of $\hat{A} = PAP^{-1}$ are of full rank, then A can be stabilized by scaling. 
	
	A more restrictive version of this condition which gives a simple way to construct the coordination matrix is given in the following theorem.
	\begin{theorem} \label{th:fullycoordN}
		If all leading principal sub-matrices of $A$ are of full rank, then there exists a diagonal matrix $C = diag(c_1,c_2,..,c_N) \in \mathbb{R}^{I \times I}$ that stabilizes $A$ (i.e. $CA$ is Hurwitz).
	\end{theorem}
	\begin{proof}
		Indeed, it then suffices to choose $c_1,c_2,...,c_I$ sequentially such that $(-1)^i c_1 \dots c_i \det([A]_{i,i}) > 0$ for $i=1,\dots, I$ where $[A]_{i,i}$ is the submatrix of $A$ comprised of lines 1 trough $i$ and columns 1 through $i$. This means that $\forall k = 1..I$, $(-1)^k det([CA]_{k,k}) > 0$ which implies by a known result \cite[Section 16.7]{simon1994mathematics} on negative definite matrices that all eigenvalues of $CA + (CA)^T$ are strictly negative.
	\end{proof}
	Later works have extended the Fisher and Fuller condition to more general cases \cite{roy2006some}.
	
\section{Coordination as a convex optimization problem} \label{sec:coordmeth}
	This section aims to formalize the problem of finding a coordination matrix $C$ such that \eqref{eq:son_linear_ode_coord} is stable when \eqref{eq:son_linear_ode} is not. We begin by recalling a sufficient condition for stability mentioned in Section \ref{subsec:stabchar}, declining it for the linear case.
\begin{theorem} \label{coordequilpoint}
Suppose there exists a $I \times I$ matrix $C$ verifying \\
\begin{equation} \label{condition i}
(CA)^T+CA \prec 0,
\end{equation}
then $A$ and $C$ are invertible, $\theta^*$ is the only equilibrium point of \eqref{eq:son_linear_ode_coord} and it is globally asymptotically stable.
\end{theorem}
\begin{proof}
	If $CA+(CA)^T \prec 0$, then CA is invertible, and so the equation $CF(\theta) = CA(\theta - \theta^*) = 0$ has a unique solution which is $\theta^*$.

	The global asymptotic stability is obtained from Theorem \ref{th:diagstrictconc}, since condition (\ref{condition i}) implies diagonal strict concavity.
\end{proof}

Note that $\theta^*$ is also an equilibrium point of \eqref{eq:son_linear_ode}. In addition to the constraint \eqref{condition i} we need to consider an additional constraint which is related to the capability of the different \ac{SON} entities to exchange information. For example, if two \acp{SON} $i$, $j$ are located at different \acp{BS} in a \ac{LTE} network without a X2 interface between them, then the element $C_{i,j}$ in matrix $C$ must be equal to $0$. On the other hand, if  $C_{i,j}\neq0$, then updating the parameter $\theta_i$ requires the value of $f_j(\theta)$, so we have to be sure that this information can be made available. Typically in a network for example, this relates to interfaces that exist between \acp{BS}, so the system constraints will be mapped from the network architecture. To add this constraint, we define a set of system-feasible matrices in $\mathbb{R}^{N \times N}$, called $\cal C$ reflecting our system constraints.

	Denote the two constraints mentioned above stability and connectivity constraints. These two constraints may be verified by a large number of matrices, and the one with the best convergence properties is sought. From convex optimization theory, we know that iterative algorithms converge faster when their condition number gets lower \cite{pyzara2011influence}. Indeed, the solution of the system of \acp{ODE} $\dot{x} = CAx$ is in fact $x(t) = e^{t CA} x_0$. The exponential of a matrix is defined using the power series. The solution is then rewritten as
	\begin{equation*}
		x(t) = \left(\sum_{k=0}^\infty \frac{t^k}{k!} (CA)^k \right) x_0
	\end{equation*}
Now if we choose $x_0$ as an eigenvector of $CA$ with the eigenvalue $\lambda_0$, we can see that $x(t) = \left(\sum_{k=0}^\infty \frac{t^k}{k!} \lambda_0^k \right) x_0 = e^{\lambda_0 t} x_0$. Now this is true for all the eigenvectors of the matrix $CA$ so that for a random starting point $x_0$, a lower condition number will ensure a better convergence as the speed of convergence will be homogeneous across the eigenspaces.

	So without constraints, the best coordination matrix would be $-A^{-1}$, leading to a diagonal matrix $CA = -I_N$ which gives us a system with the lowest condition number i.e. 1. When taking the constraints into account, we formulate the convex optimization problem as minimizing the distance, defined in terms of the Frobenius norm, between the coordination matrix $C$ and $-A^{-1}$:

\begin{equation} \label{cvx}
	\begin{aligned}
		& \textit{minimize  } \| C + A^{-1} \|_F \\
		& \textit{s.t. } (CA)^T + CA \prec 0; C \in \cal C
	\end{aligned}
\end{equation}
where $\| . \|_F$ is the Frobenius norm defined for an $\mathbb{R}^{m \times n}$ matrix $M$ as
\begin{equation}
\|M\|_F = \sqrt{\sum_{i=1}^m \sum_{j=1}^n |M_{i,j}|^2} = \sqrt{Tr(M^T M)} = \sqrt{\sum_{i=1}^{\text{min}(m,n)} \sigma_i^2}
\end{equation}
with the $\sigma_i$ being the singular values of $M$. It is noted that the Frobenius norm is often used in the literature for finding a preconditioner that improves the convergence behavior of iterative inversion algorithms \cite{fronormprecond}.

The stability constraints are expressed in the form of \acp{LMI}. \acp{LMI} are a common tool used in control theory for assessing stability. Solving convex optimization problems with \ac{LMI} constraints is a tractable problem and fast solvers are available \cite{LMIControlBoyd}.

From the implementation point of view, the coordination process can be performed in two steps as follows. In the first step, a centralized coordination server gathers and processes data to derive the matrix $A$, performs the optimization problem \eqref{cvx} to obtain the coordination matrix $C$,
and sends each line of the matrix $C$ to the corresponding \ac{SON} entity. This step is performed off-line.
The second step is the on-line control process where each \ac{SON} performs the coordinated control, while satisfying the connectivity constraints, by using  the appropriate line of matrix $C$.

\section{Application to wireless networks}\label{sec:application}
In this section we illustrate instability and coordination in the context of wireless networks using two examples. We first show that instability occurs even in simple models with as few as two \acp{SON} in parallel. Then we apply the coordination to a use case involving 3 \acp{SON} deployed in several \acp{BS} in a \ac{LTE} network.
\subsection{Admission control and resource allocation}
\subsubsection{Model}
	We consider a \ac{BS} in downlink, serving elastic traffic. Users enter the network according to a Poisson process with arrival rate $\lambda$, to download a file of exponential size $\sigma$, with $\expec{\sigma} < +\infty$. The \ac{BS} has $x_{max}$ parallel resources available, and we write $x \in [0,x_{max}]$ the amount of resources used. We ignore the granularity of resources, either assuming that there are a large number of resources or using time sharing, using each resource a proportion $\frac{x}{x_{max}}$ of the time. Depending on the access technology, resources can be: codes in \ac{CDMA}, time slots in \ac{TDMA}, time-frequency blocks in \ac{OFDMA} etc. When a user is alone in the system, his data rate is $R x$. Users are served in a processor sharing manner (for instance through Round Robin scheduling): if there are $n$ active users, each user has a throughput of $\frac{x R}{n}$. Admission control applies. We define $\beta \geq 0$ a blocking threshold and the probability of accepting a new user when $n$ users are already present in the system is $\phi(n - \beta)$ with $\phi:\RR \to [0,1]$ a smooth, strictly decreasing function and $\phi(n) \tends{n}{+\infty} 0$. We choose $\phi$ as a logistic function for numerical calculations:
\begin{equation} 
	\phi(n) = \frac{1}{1 + e^{n}}.
\end{equation}
	
	Define $n(t)$ the number of active users in the system at time $t$, then $n(t)$ is a continuous time Markov chain. $n(t)$ is ergodic because the probability of accepting a new user goes to $0$ as $n \to \infty$. We define the load:
\begin{equation}
	\rho(x) = \frac{\lambda \expec{\sigma} }{x R}.
	\label{eq:load_defintion}
\end{equation}
We write $\pi$ the stationary distribution of $n(t)$. $n(t)$ is reversible, and $\pi$ can be derived from the detailed balance conditions:
\begin{equation}
	\pi(n,x,\beta) = \frac{ \rho(x)^n \prod_{k=0}^{n-1} \phi(k-\beta) }{\sum_{n \geq 0}  \rho(x)^n \prod_{k=0}^{n-1} \phi(k-\beta) }.
	\label{eq:statonary_distribution}
\end{equation}
Using Little's law, the mean file transfer time is given by the average number of active users divided by the arrival rate:
\begin{equation}
	T(x,\beta) =\frac{1}{\lambda} \sum_{n \geq 0} n\pi(n,x,\beta).
	\label{eq:file_transfer}
\end{equation}
Let $R_{min}$ a minimal data rate required to ensure good \ac{QoS}. We say that there is an outage in a state of the system if users have a throughput lower than $R_{min}$. When there are $n$ active users in the system, outage occurs \ac{iff}:
\begin{equation}
	n > \frac{xR}{R_{min}}.
	\label{eq:outage_cond}
\end{equation}
The outage probability is then:
\begin{equation}
	O(x,\beta) = \sum_{n \geq 0} \pi(n,x,\beta) \indic_{ (0,+\infty) }\left(n - \frac{xR}{R_{min}}  \right).
	\label{eq:outage}
\end{equation}
In this model, $x \to O(x,\beta)$ is not smooth, which is why we introduce the smoothed outage $\tilde{O}$: 
\begin{equation}
	\tilde{O}(x,\beta) = \sum_{n \geq 0} \pi(n,x,\beta)  \psi \left(n - \frac{xR}{R_{min}}\right).
	\label{eq:outage_smooth}
\end{equation}
with $\psi$ a smooth function approximating  $\indic_{(0,+\infty)}$. We also choose $\psi$ as a logistic function for numerical calculations.

 This queuing system is controlled by two mechanisms, and that control occurs on a time scale much slower than the arrivals and departures of the users, so that the mean file transfer time and outage probability are relevant performance metrics, and can be estimated from (noisy) measurements. The mechanisms are:
\begin{itemize}
	\item \emph{Resource allocation:} a mechanism adjusts the amount of used resources to reach a target outage rate. Such mechanisms have been considered in green networking when a \ac{BS} can switch off part of its resources in order to save energy. 
		
	Another application is interference coordination: using more resources will create inter-cell interference in neighboring \acp{BS} and degrade their \ac{QoS}. Hence \acp{BS} should use as little resources as possible, as long as their target \ac{QoS} is met.
	\item \emph{Admission control: } another mechanism adjusts the admission control threshold to reach a target file transfer time. In particular, it is noted that without admission control, the mean file transfer time becomes infinite in overload.
\end{itemize}
It is noted that $x \to \tilde{O}(x,\beta)$ is strictly decreasing and \newline $\beta \to T(x,\beta)$ is strictly increasing. Using the notations of Section \ref{sec:parallel_control_loops}, we have $I = 2$ control loops, $\theta_1 \equiv x$, $\theta_2 \equiv \beta$, $f_1 \equiv \tilde{O}$, $f_2 \equiv -T$. Consider $\theta^* = (x^*,\beta^*)$ an operating point. The stability in the neighborhood of $(x^*,\beta^*)$ can be calculated. The system will fail to converge to the desired operating point as long as the Jacobian matrix has a negative determinant, so that there are two eigenvalues of opposite sign:
	\begin{equation}
		-\frac{\partial \tilde{O} }{\partial x} \frac{\partial T}{\partial \beta}(x^*,\beta^*) +  \frac{\partial \tilde{O} }{\partial \beta } \frac{\partial T}{\partial x}(x^*,\beta^*) < 0
		\label{eq:instability_condition}
	\end{equation}
\subsubsection{Results}
	We now evaluate the stability region of the system numerically by checking condition~\eqref{eq:instability_condition} for various operating points. We choose the following parameter values: $\lambda = 0.5 users/s$, $\expec{\sigma} = 10 Mbits$, $R = 15 Mbits/s$, $R_{min} = 2 Mbits/s$, $x_{max} = 1$. Figure~\ref{fig:stability_region} presents the results. In the white region the system is stable, and in the gray region it is unstable. Even in such a simple setting with $1$ \ac{BS} and $2$ \ac{SON} mechanisms, instability occurs for a large set of operating points.
	\begin{figure}[htbp]
		\begin{center}
			\includegraphics[width=\figsize]{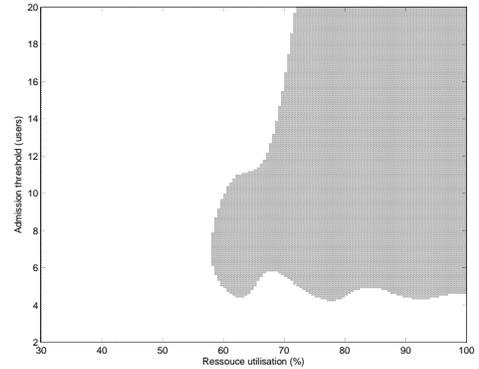}
		\end{center}
		\caption{Stability region of the system}
		\label{fig:stability_region}
	\end{figure}
	
\subsection{\ac{SON} Coordination use case}\label{useCase}
\subsubsection{System Model} \label{systmodel}
Consider three SON mechanisms deployed in \acp{BS} of a \ac{LTE} network: blocking rate minimization, outage probability minimization and load balancing \ac{SON} as presented in \cite{CombesInfocom2012}. We focus on downlink ftp type traffic model in which each user enters the network, downloads a file from its serving cell and then leaves the network.
\paragraph{Load balancing}
The \ac{SON} adjusts the \ac{BS}'s pilot powers in order to balance the loads between neighboring cells.  The corresponding ODE is given by
\begin{equation}
\dot{P_s} = P_s(\rho_1(\textbf{P})-\rho_s(\textbf{P})), \, \, \forall s=1...I
\end{equation}
where $\textbf{P} \in \mathbb{R}^I$ is the vector of \acp{BS}' pilot powers, and $\rho$ - their corresponding loads. This SON converges to a set on which all loads are equal as shown in [3, Theorem 4]. We use an equivalent formulation in order to update the pilots directly in dB:
\begin{equation}
\dot{P}_{{\text{dB}}_s} = \rho_1(\textbf{P})-\rho_s(\textbf{P}) \, \, \forall s=1...I.
\end{equation}
where ${P_{\text{dB}}}_s$ is the pilot power of \ac{BS} $s$ in decibels.
\paragraph{Blocking rate minimization}
This \ac{SON} adjusts the admission threshold in order to reach a given blocking rate target $\bar{B} > 0$. Denote by $x_s \in \mathbb{R}^+$ a real value whose floor (i.e. the smaller integer part) equals the effective admission threshold of \ac{BS} $s$, i.e. a new user finding the cell with $n$ users is blocked with probability $P(n)$, where $P(n) \to 1$ when $n \to x_s$ and $P(n) \to 0$ when $n \to 0$. The update equation for the blocking rate minimization SON is
\begin{equation} \label{bcrminsa}
x_{s,t+1} = [x_{s,t} +  \epsilon_t (B_s(x_t) - \bar{B}_s + N_t)]_{[0,x_{max}]}^+
\end{equation}
where $x_t$ is the vector of the admission thresholds of all the \acp{BS} considered at time t, $x_{max}$ a sufficiently large value and $N_t$ a martingale difference noise introduced by measuring $B[x_t]$.
The equivalent \ac{ODE} is
\begin{equation} \label{bcrmin}
\dot{x}_s = B_s(x) - \bar{B}_s.
\end{equation}
$x_s \to B_s(x)$ is a decreasing function of $x_s$ and $\lim_{x_s \to \infty} B_s(x) = 0$. So for any blocking rate target $0 < \bar{B}_s < 1$, we have $B_s(0) \geq \bar{B}_s$ and there exists a finite $x_0 \in \mathbb{N}$ such that $\forall x \geq x_0; B_s(x) \leq \bar{B}$ and $\forall x \leq x_0; B(x) \geq \bar{B}$. By projecting the right hand side of \eqref{bcrminsa} on any interval containing $[0,x_0]$, we ensure that $\text{sup}_t \| x_t \| < \infty$.

	Now considering the function $V(x) = max(0,|x-x_0|-\delta)$ for $\delta$ sufficiently small, we can see that $V(.)$ is a Lyapunov function for \eqref{bcrmin}. Indeed, we have
\begin{itemize}
\item $\forall x \in [0,+\infty), \, V(x) \geq 0$.
\item $H = \{ x \in [0,+\infty), V(x) = 0  \} \neq \emptyset $ because it contains $x_0$.
\item $\displaystyle\begin{aligned}[t]
\dot{V}(x) & = \frac{\partial V}{\partial x} \dot{x} \\ & = 
\begin{cases} 
	-(B(x)-\bar{B}) &\mbox{if } x < x_0-\delta \\
	B(x)-\bar{B} &\mbox{if } x > x_0+\delta \\
	0 &\mbox{if } x \in [x_0-\delta,x_0+\delta]
\end{cases} \\
& \leq 0.
\end{aligned}$
\item $V(x) \to +\infty$ when $x \to +\infty$.
\end{itemize}
This implies that $H$ is globally asymptotically stable for \eqref{bcrmin} (see Appendix \ref{app:odes}).

\paragraph{Outage Probability Minimization}
The aim of this \ac{SON} mechanism is to adjust the transmit data power in order to reach a target outage probability. The outage probability considered is expressed as
\begin{equation} \label{cov_prob}
K_s = \frac{1}{|Z_s|} \int_{Z_s} \mathds{1}_{\{R_s(r) \geq R_{\textit{min}}\}} (r) dr
\end{equation}
where $Z_s$ is the area covered by \ac{BS} $s$, $R_\textit{min}$ a minimum data rate and $R_s(r)$ the peak data rate obtained at position $r$ when served by BS $s$. The \ac{SA} algorithm modeling the actual control loop is
\begin{equation}\label{covmaxsa}
D_s[k+1] =  D_s[k] - \epsilon_k(K_s(\textbf{D}[k]) - \bar{K} + N_k^s)
\end{equation}
where $N_k^s$ is a martingale difference noise and $D_s$ is the transmit data power of \ac{BS} $s$. The limiting \ac{ODE} representing the mean behaviour of \ac{SA} \eqref{covmaxsa} is then
\begin{equation} \label{covmax}
\dot{D}_s = - (K_s(D) - \bar{K}).
\end{equation}
This \ac{ODE} is stable if there exists an admissible data power $D_s^*$ such that $K_s(D_s^*) = \bar{D_s}$. Indeed, $(K_s(.)-\bar{K})^2$ would then be a Lyapunov function for \eqref{covmax} since $\frac{\partial K_s}{\partial D_s} > 0$. As a consequence, the \ac{SA} \eqref{covmaxsa} converges to invariant sets of \eqref{covmax}, which means that the mechanism is standalone-stable.

\subsubsection{Numerical Results}
Consider a hexagonal network with 19 cells with omni-directional antennas as shown in \figurename \ref{cell-layout}. A wrap-around model is used to minimize truncation effects of the computational domain. It is achieved by surrounding the original network with 6 of its copies while performing the simulation within the original 19 cells.
\begin{figure}[!ht]
\centering
\includegraphics[width=3.2in]{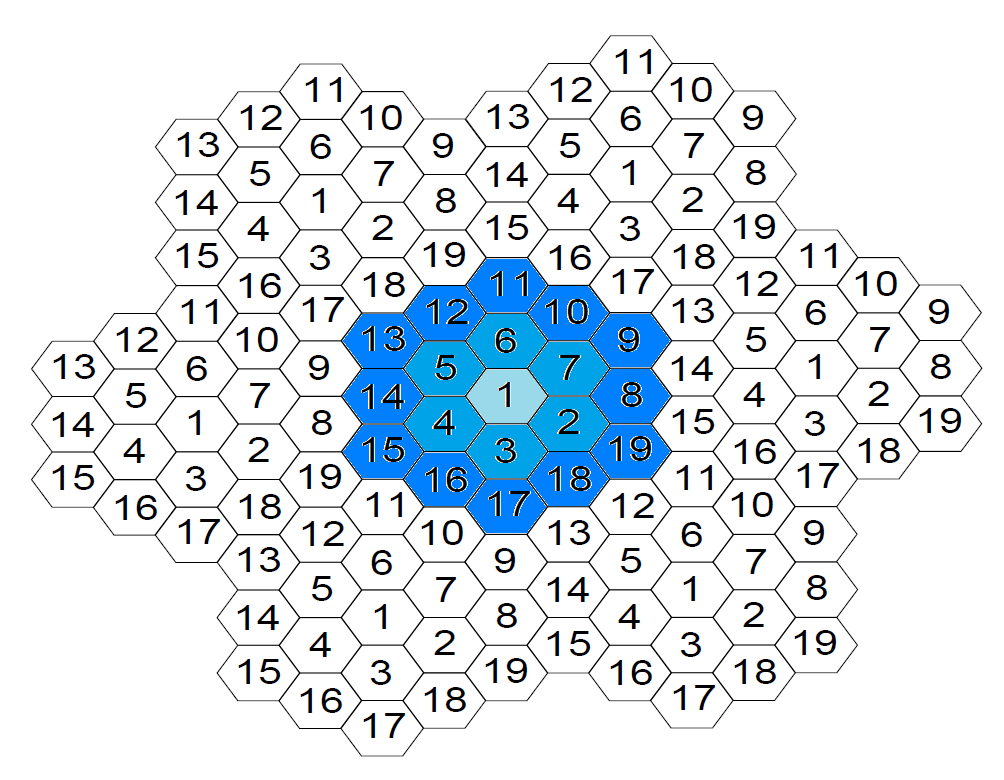}
\caption{19 cells hexagonal network with wrap-around}
\label{cell-layout}
\end{figure}
\begin{table}[!t]
\renewcommand{\arraystretch}{1.3}
\caption{Network and Traffic characteristics}
\label{tab:params}
\centering
\begin{tabular}{|c|c|}
\hline
\multicolumn{2}{|c|}{Network parameters} \\
\hline
Number of stations & 19 (with wrap-around) \\
\hline
Cell layout & hexagonal omni \\
\hline
Intersite distance & 500 m \\
\hline
Bandwidth & 20MHz \\
\hline
\multicolumn{2}{|c|}{Channel characteristics} \\
\hline
Thermal noise & -174 dBm/Hz \\
\hline
Path loss (d in km) & 128 + 36.4 $\log_{10}(d)$ dB \\
\hline
\multicolumn{2}{|c|}{Traffic characteristics} \\
\hline
Arrival rate & 40 users/s \\
\hline
Service type & FTP \\
\hline
Average file size & 10 Mbits \\
\hline
Hotspot additional arrival rate & 2 users/s \\
\hline
Hotspot position & center of \ac{BS} 1 cell \\
\hline
Hotspot diameter & 330 m \\
\hline
\multicolumn{2}{|c|}{Simulation parameters} \\
\hline
Spatial resolution & 20 m x 20 m \\
\hline
Time per iteration & 6 s \\
\hline
Minimun SINR for coverage & 0 dB \\
\hline
Target outage probability & 82\% \\
\hline
Target blocking rate & 2\% \\
\hline
\end{tabular}
\end{table}

	Table \ref{tab:params} lists the parameters used in the simulations including the environment, the network, the numerical simulation parameters and the KPIs' targets used by the SON mechanisms. The users arrive in the network according to a Poisson process with a certain arrival rate given in Table \ref{tab:params}. And a hotspot is placed at the center of the network with an additional arrival rate also given in Table \ref{tab:params}. The hotspot provides initially unbalanced loads in the network, which is of interest for the load balancing \ac{SON}.

	We activate the three SONs in each of the 7 \acp{BS} located at the center of the map and observe the stability of the \acp{SON}, with and without the coordination mechanism.
	
	The first step consists in computing the stability matrix of the linearized control system. This step is performed by evaluating the closed-form formulas of the corresponding \acp{KPI} and then computing the partial derivatives using finite differences. By choosing an adequate step size, this method yields very accurate results. However, closed-form expressions of the \acp{KPI} do not always exist, in which case estimations of the \acp{KPI} would be used instead, based on measurements from each user that arrive in the network. The stability matrix obtained through linearization already reveals instability since not all of its eigenvalues are negative, and hence the coordination step is inevitable.

	We plot the \acp{KPI} evolutions for the coordinated (in blue) and non-coordinated (in red) systems (Figures 3 to 5). The coordinated system clearly performs better. The loads are lower, the blocking rates too. The outage probabilities in the non-coordinated system diverge. The most loaded \ac{BS} outage probability is near zero while that of the others reach the maximum. This is because the decrease in cell size of the most loaded \ac{BS} is not followed by a decrease of its traffic power, causing more interference on its neighbors which have increased their cell size. The goals of each \ac{SON} are more or less achieved in the coordinated system. And the loads are balanced very quickly (less than 10 minutes), this shows the usefulness of distributed \acp{SON} over centralized ones.
	
	In \figurename \ref{coordk2}, we can see that the outage of \ac{BS} 1 in the coordinated network is good but off the target we had set for the outage \ac{SON}. This is a consequence of putting together many types of \acp{SON} which are not homogeneous. If needed, the \ac{SON} entities can be harmonized by using some weights. We now investigate the impact of the weights on the stationary \acp{KPI} of the system.
	
	Figures \ref{op_point} and \ref{op_point2} compare the final values of the \acp{KPI} of the coordinated (in blue) and non-coordinated (in red) systems when they reach their permanent state for different weight vectors. For equal weight across all SONs, we can see in \figurename \ref{op_point} that the system is rather drifted towards balancing the loads as not all the \acp{BS} reach their outage target. Actually traffic power of \ac{BS} 1 (the most loaded one) increases to absorb more traffic while its cell size is reduced leading to a smaller outage probability.

	We then give more importance to outage target in the \acp{SON} by increasing 20 times its weight. In Figure \ref{op_point2}, we plot the numerical results obtained. The outages are closer to their goals but the system does not balance the loads anymore. This shows us that with the coordination, a compromise is made. This compromise can be changed by adjusting the weights of the \acp{SON} to reflect the policies or the priorities of the network operator.

\begin{figure}[!ht]
\centering
\includegraphics[width=3in]{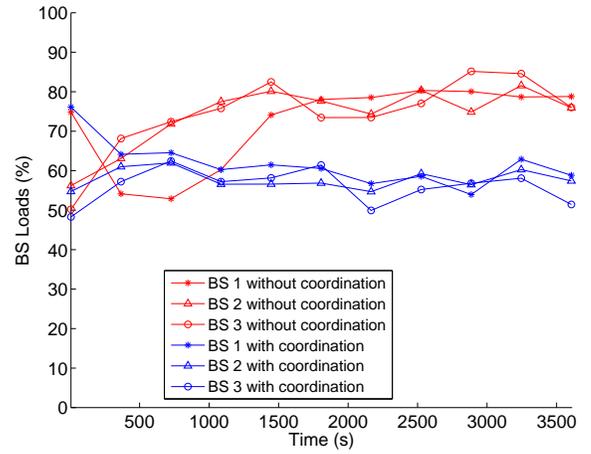}
\caption{Impact of Coordination on Loads}
\label{coordk1}
\end{figure}

\begin{figure}[!ht]
\centering
\includegraphics[width=3in]{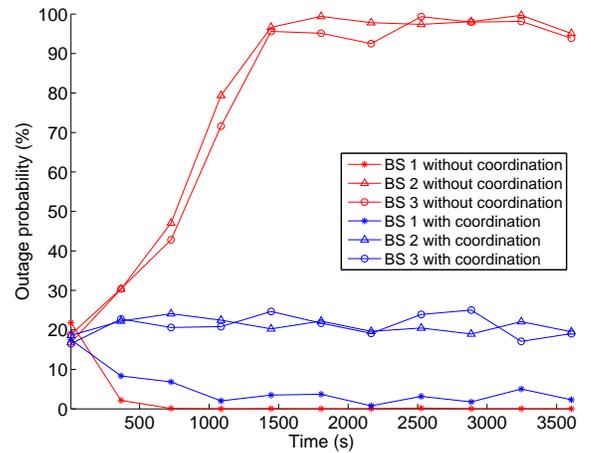}
\caption{Impact of Coordination on Coverage Probabilities}
\label{coordk2}
\end{figure}

\begin{figure}[!ht]
\centering
\includegraphics[width=3in]{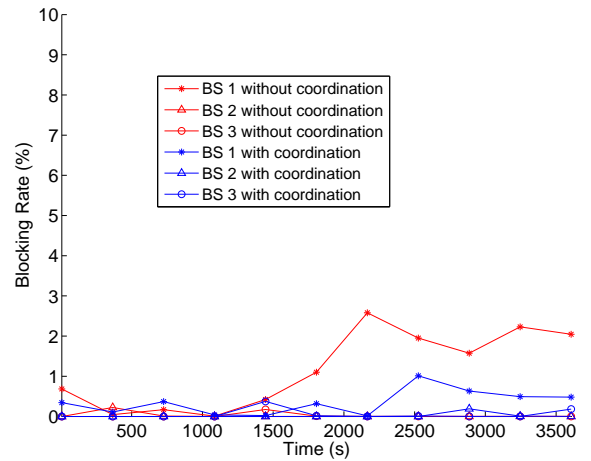}
\caption{Impact of Coordination on Blocking Rates}
\label{coordk3}
\end{figure}

\begin{figure}[!ht]
\centering
\includegraphics[width=3in]{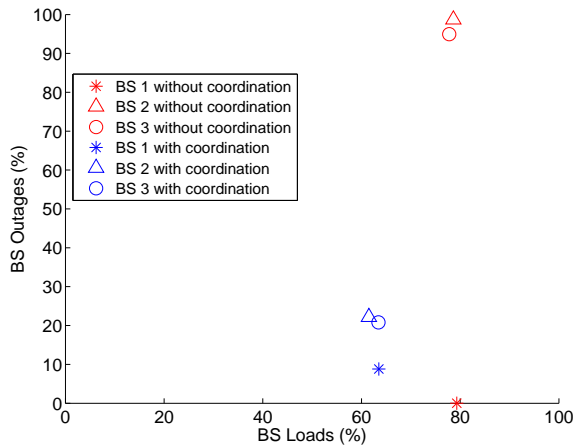}
\caption{Stationary KPIs with with all \acp{SON} equally weighted}
\label{op_point}
\end{figure}

\begin{figure}[!ht]
\centering
\includegraphics[width=3in]{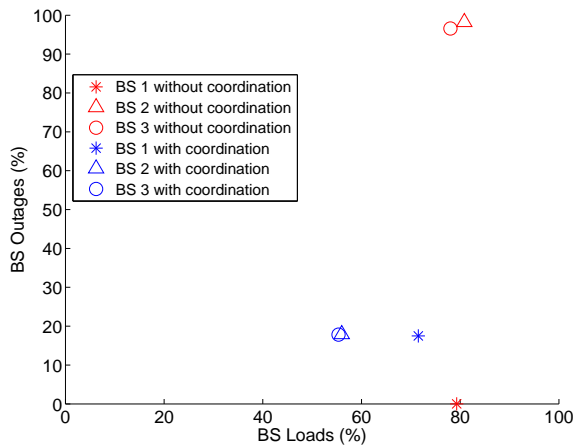}
\caption{Stationary KPIs with outage probability prioritized}
\label{op_point2}
\end{figure}

\section{Concluding Remarks} \label{sec:conclusion}
In this paper we have studied the problem of coordinating multiple \ac{SON} entities operating in parallel.  Using tools from control theory and Lyapunov stability, we have proposed a coordination mechanism to stabilize the system. The problem of finding a coordination matrix has been formulated as a convex optimization problem with \acp{LMI} constraints which ensures that the system of \ac{SON}s remain distributed. The coordination can be implemented in a distributed fashion so it is scalable with respect to the number of \acp{SON}. We have also shown that the coordination solution remains valid in the presence of measurement noise, using stochastic approximation. Instability in the context of wireless networks have been illustrated with an example. We have shown that even for two control loops, instability can occur, and the influence of network geometry has been investigated. A practical use case of the coordination method has been presented in a LTE network implementing three distributed SON functions deployed in several base stations. The coordination reveals to be necessary in this use case since the system of \ac{SON}s has shown to be unstable without coordination. The different \ac{SON} entities achieve their respective goals when coordinated and the network is stabilized. This use case has also shown that in spite of the linear control assumption, the method remains effective when applied to \ac{SON} functionalities that are not linear in general.

\appendices
\section{Proof of Lemma \ref{lem:routh}}
\begin{proof}
	We denote by $P(\lambda) = \det(\lambda I_{I} - M )$ the characteristic polynomial of $M$. $M$ is a Hurwitz matrix \ac{iff} $P$ has no roots with positive real part. For $I=2$, \begin{equation}P(\lambda) = \lambda^2 - \tr(M)\lambda + \det(M) .\end{equation} By the Routh-Hurwitz theorem, $P$ has no roots with positive real part iff:
	\begin{equation} -\tr(M) > 0 \;,\; -\tr(M)\det(M) > 0 \end{equation} which proves the result. For $I=3$, \begin{equation*}P(\lambda) = \lambda^3 - \tr(M)\lambda^2 - \frac{\lambda}{2} ( \tr(M^2) - \tr(M)^2 ) - \det(M).\end{equation*} By the Routh-Hurwitz theorem, $P$ has no roots with positive real part iff: $-\tr(M) > 0$, $-\det(M) > 0$ and:
	\begin{equation} \frac{\tr(M)}{2} (  \tr(M^2) - \tr(M)^2  ) + \det(M) > 0. \end{equation}
	For $I=3$, we have that: 	\begin{equation}\tr(M^2) - \tr(M)^2 = - 2 \det(M) \tr(M^{-1}),\end{equation}
	so that the last condition reduces to \begin{equation} \tr(M) \tr(M^{-1}) - 1 > 0\end{equation} proving the result.
\end{proof}

\section{Asymptotic behavior of ODEs}\label{app:odes}
	Consider the \ac{ODE}:
\begin{equation}
	\dot{\theta} = F(\theta),
\end{equation}
	which we assume to have a unique solution for each initial condition defined on $\RR^+$. We write $\Phi(t,\theta(0))$ the value at $t$ of the solution for initial condition $\theta(0)$. We denote by $d_{{\cal U}}(\theta) = \infu{u \in {\cal U}} \norm{ \theta - u}$ the distance to set ${\cal U}$. We say that ${\cal U}$ is \emph{invarient} if $\theta(0) \in {\cal U}$ implies $\Phi(t,\theta(0)) \in {\cal U} $ , $t \in \RR^+$. We say that ${\cal U}$ is \emph{Lyapunov stable} if for all $\delta_1 > 0$ there exists $\delta_2 > 0$ such that $d_{{\cal U}}(\theta(0)) \leq \delta_2$ implies $d_{{\cal U}}(\Phi(t,\theta(0))) \leq \delta_1$ , $t \in \RR^+$. A compact invariant set ${\cal U}$ is an \emph{attractor} if there is an open invariant set ${\cal A}$ such that $\theta(0) \in {\cal A}$ implies $d_{{\cal U}}(\Phi(t,\theta(0))) \tends{t}{+\infty} 0 $ . ${\cal A}$ is called the \emph{basin of attraction}.  A compact invariant set ${\cal U}$ is \emph{locally asymptotically stable} if it is both Lyapunov stable and an attractor. If its basin of attraction ${\cal A}$ is equal to the whole space then ${\cal U}$ is \emph{globally asymptotically stable}. Asymptotic stability is often characterized using \emph{Lyapunov functions}. A positive, differentiable function $V:\RR^I \to \RR$, is said to be a Lyapunov function if  $t \mapsto V(\Phi(t,\theta(0)))$ is decreasing, and strictly decreasing whenever $V(\Phi(t,\theta(0))) > 0$. Then the set of zeros of $V$ is locally asymptotically stable. If we add the condition $V(\theta) \tends{\norm{\theta}}{+\infty} +\infty$, then we have global asymptotic stability.

\section{Linear ODEs}\label{app:linear_ode}
	Consider the \ac{ODE}:
\begin{equation}\label{eq:son_linear_ode_app}
	\dot{\theta} = A (\theta - \theta^*).
\end{equation}
Its solution has the form:
\begin{equation}\label{eq:son_linear_ode_sol}
	\theta(t) =  \theta^* +  e^{t A} (\theta(0) - \theta^*).
\end{equation}
	We denote by $\prec$ positive negativity for symmetric matrices. $\theta^*$ is asymptotically stable \ac{iff} all the eigenvalues of $A$ have a strictly negative real part. Alternatively, asymptotic stability applies \ac{iff} there exists $0 \prec X$ such that $A^T X +  X A \prec 0$. In this case, $V(\theta)= (\theta - \theta^*)^T  X (\theta - \theta^*)$ is a Lyapunov function for the \ac{ODE}. The reader can refer to \cite{LMIControlBoyd} for the linear matrix inequality approach to stability.
	
\section{Definition of diagonal strict concavity}
	Diagonal strict concavity is a property introduced in \cite{RosenConvexGame} for analyzing equilibrium of n-person games. Consider $I$ functions $\theta \to g_i(\theta)$ defined on a convex closed bounded set $S \subset \mathbb{R}^I$, and $w_i, i=1,\dots ,I$ some real positive constants. And denote by
	\begin{equation}
		JG(\theta) = \begin{bmatrix}
						w_1 \nabla_1 g_1(\theta)\\
						.\\
						.\\
						.\\
						w_I \nabla_I g_I(\theta)
					 \end{bmatrix}.
	\end{equation}
	We say that $G(\theta)=\sum_{i=1}^I w_i g_i(\theta)$ is diagonally strictly concave for $\theta \in S$ if for every $\theta_0,\theta_1 \in S$ we have
	\begin{equation}
		(\theta_0-\theta_1)^T JG(\theta_0)+(\theta_0-\theta_1)^T JG(\theta_1) > 0
	\end{equation}

\section{Martingales}
	Martingales are commonly used to characterize noise in stochastic approximation algorithms \cite{Borkar,Kushner}. We hereby give a succinct definition of martingales and martingale differences along with an insight in why they are useful.

	Let $(\Omega,\cal F, P)$ denote a probability space, where $ \Omega $ is the sample space, $\cal F$ a $\sigma$-algebra of subsets of $\Omega$, and $P$ a probability measure on $(\Omega,\cal F)$. Let $\{ M_n\}$ be a sequence of real-valued random variables defined on $(\Omega,\cal F)$. If $\mathbb{E}(|M_n|) < \infty$ and
	\begin{equation}
		\mathbb{E}(M_{n+1} | M_i,i \leq n) = M_n
	\end{equation}
	then $\{ M_n\}$ is a martingale sequence. In this case, the sequence $N_n = M_n - M_{n-1}$ is a martingale difference sequence.
	
	An important result on martingales is the martingale convergence theorem which proves that martingale sequences converge with probability 1. This result is useful to characterize convergence of \ac{SA} algorithms which model the noise as martingale differences (see \cite{Borkar} and \cite{Kushner}).

\bibliographystyle{IEEEtran}
\bibliography{main}

\end{document}